\documentclass[journal,letterpaper]{IEEEtran}
\usepackage{amsmath, amssymb, graphicx, cite, color}
\usepackage{algorithmic}

\newtheorem{theorem}{Theorem}
\newtheorem{lemma}{Lemma}
\newtheorem{proposition}{Proposition}
\newtheorem{corollary}{Corollary}
\newtheorem{remark}{Remark}

\newtheorem{definition}{Definition}

\newcommand{\tr}{\text{Tr}}
\newcommand{\be}{\begin{eqnarray}}
\newcommand{\ee}{\end{eqnarray}}

\newcommand{\EX}{\mathbf{E}}
\newcommand{\PR}{\mathbf{P}}

\newcommand{\mc}{\mathcal}
\newcommand{\mbf}{\mathbf}

\newcommand{\GI}{\mathsf{GI}}
\newcommand{\Mq}{\mathsf{M}}

\newcommand{\cH}{\mathcal{H}}
\newcommand{\cX}{\mathcal{X}}

\newcommand{\cN}{\mathcal{N}}
\newcommand{\cS}{\mathcal{S}}
\newcommand{\qN}{B} 

\IEEEoverridecommandlockouts 

\title{The Classical Capacity of Additive Quantum Queue-Channels}
\author{Prabha Mandayam, Krishna~Jagannathan,
        and~Avhishek Chatterjee
        \thanks{P. Mandayam is with the Department of Physics, IIT Madras, Chennai, India (e-mail: prabhamd@iitm.ac.in).}
\thanks{K.~Jagannathan and A.~Chatterjee are with the Department of Electrical Engineering, IIT Madras, Chennai , India (e-mail: \{krishnaj,avhishek\}@ee.iitm.ac.in). AC acknowledges the Department of Science and Technology, Govt. of India for its support through  SERB/SRG/2019/001809 and INSPIRE/04/2016/001171.}%
\thanks{The material in this paper was presented in part at the 2019 IEEE 20th International Workshop on
Signal Processing Advances in Wireless Communications (SPAWC) \cite{8815462}, \textcolor{black}{and the National Conference on Communications 2019 \cite{JagannathanCM2019}.}}
}

\begin{document}

\maketitle

\begin{abstract}
We consider a setting where  a stream of qubits is processed {sequentially}. We derive fundamental limits on the rate at which classical information can be transmitted using qubits that decohere as they wait to be processed. Specifically, we model the sequential processing of qubits using a single server queue, and derive  expressions for the classical capacity of such a quantum `queue-channel.' Focusing on two important noise models, namely the \emph{erasure channel} and the \emph{depolarizing channel}, we obtain  explicit single-letter capacity formulas in terms of the stationary waiting time of qubits in the queue. Our capacity proof also implies that a `classical' coding/decoding strategy is optimal, i.e., an encoder which uses only orthogonal product states, and a decoder which measures in a fixed product basis, are sufficient to achieve the classical capacity of both queue-channels. Our proof technique for the converse theorem generalizes readily --- in particular,  whenever the underlying quantum noise channel is \emph{additive}, we can obtain a single-letter upper bound on the classical capacity of the corresponding quantum queue-channel. More broadly, our work begins to quantitatively address the impact of decoherence on the performance limits of quantum information processing systems.

\end{abstract}


\section{Introduction}
\label{sec:intro}
Unlike classical bits, quantum bits (or qubits) undergo rapid \emph{decoherence} in time, due to certain unavoidable physical interactions with their environment~\cite{NCBook}. Information stored in a qubit may be completely or partially lost as the qubit decoheres. The nature and the rate of the noise process  depends on the particular physical implementation of the quantum state, as well as other factors such as the environment and  temperature.  For example, superconducting Josephson junction based qubits have average coherence times that are typically of the order of a few tens of microseconds~\cite[Table 2]{wendin2017quantum}. Decoherence or noise poses a major challenge to the scalability of quantum information processing systems --- therefore, it is imperative to obtain a quantitative understanding of the impact of decoherence on the performance limits of quantum information processing systems. 

In this paper, we consider a setting where a stream of qubits is processed \emph{sequentially}. {There are several possible scenarios where qubits may wait to be processed. For example, in the context of quantum computation, subsets of qubits are often `idling' in a quantum circuit, waiting their turn while other qubits are undergoing gate operations. Such idle qubits undergo decoherence as they wait to get processed, and the resulting errors are modelled as `storage errors' in the context of quantum fault tolerance~\cite{aliferis_FT05}.}

{In the context of quantum communication, such a sequential processing of qubits arises naturally when we study consider the problem of entanglement distribution over quantum networks~\cite{KW_2019}. In this case, classical information is encoded in pairs of entangled photons which are transmitted over lossy fibre. To sustain the entanglement over long distances, the photons undergo an entanglement swap operation at quantum repeaters that are placed at intermediate distances. Every photon that arrives at a repeater has to wait for its partner to arrive in order to be processed, in which time they may get erased or depolarized, depending on the nature of communication channel as the quantum memory used to realise the quantum repeater~\cite{qrepeater_RMP11}.} We derive fundamental limits on the rate at which classical information can be transmitted using qubits that \emph{decohere as they wait to be processed.}
 
To be more precise, we model the sequential processing of a stream of qubits using a \emph{single server queue}~\cite{Kleinrock1975_I}. The qubits arrive to be processed at a `server' according to some stationary point process. \textcolor{black}{For example, in the context of quantum communication, qubits are generated by optical sources which have an inherent randomness due to the underlying physical processes. The commonly used heralded single-photon sources, for example, rely on a nonlinear physical interaction which is probabilistic in nature~\cite{felinto2006}.} The server processes the qubits at a fixed average rate. The qubits undergo decoherence (leading to errors) as they wait to be processed, and the probability of error/erasure of each qubit is modeled as a function of the time spent in the queue by \emph{that} qubit. After the processing completes, the qubits are measured and interpreted as classical bits. We call this system a `quantum queue-channel' and characterise the classical capacity of such a quantum queue-channel. 

 \subsection{Related Work}
An information theoretic notion of reliability of a queuing system with state-dependent  errors was introduced and studied in \cite{ChatterjeeSV2017}, where the authors considered queue-length dependent errors motivated mainly by human computation and crowd-sourcing. The classic paper of Anantharam and Verd\'{u} considered timing channels where information is encoded in the times between consecutive information packets, and these packets are subsequently processed according to some queueing discipline \cite{AnantharamV1996}. Due to randomness in the sojourn times of packets through servers, the encoded timing information is distorted, which the receiver must decode. 
In contrast to~\cite{AnantharamV1996}, \emph{we are not concerned with information encoded in the timing between packets --- in our work, all the information is in the qubits.}
 
In a recent paper \cite{JagannathanCM2019}, we considered the queue-channel problem described above, and derived the capacity of an erasure queue-channel under certain technically restrictive conditions. Specifically,  \cite{JagannathanCM2019} restricts the encoder to using only  orthogonal product states, and the decoder measures in a \emph{fixed} product basis. In this restricted setting, qubits are essentially made to behave like `classical bits that decohere,' and the underlying quantum channel effectively simulates a classical channel known as the \emph{induced classical} channel. In general, the classical capacity of the underlying quantum channel could be \emph{larger} than the capacity of the induced classical channel, because the former allows for entangled channel uses and more general (joint) measurements at the decoder. 

As an aside, we remark that the erasure queue-channel treated in \cite{JagannathanCM2019} can be used to model a multimedia-streaming scenario, where information packets become useless (erased) after a certain time. 

\subsection{Our Contributions}
In this paper, we completely characterise the classical capacity of quantum queue-channel for two important noise models, namely the \emph{erasure channel} and the \emph{depolarizing channel.} Specifically, we allow for possibly entangled channel uses by the encoder, and arbitrary measurements at the decoder. 

Obtaining the queue-channel capacity for an arbitrary noise model involves overcoming some technical challenges. First, the quantum queue-channel is non-stationary. Second, the erasure events corresponding to consecutive qubits are correlated through their waiting times, which are in turn governed by the queuing process. This leads to memory across consecutive channel uses. 

Interestingly, we note that the capacity result in \cite[Theorem 1]{JagannathanCM2019} for the induced classical channel readily offers an `achievable rate' for the quantum erasure queue-channel --- after all, any rate that is achievable with the restrictions in \cite{JagannathanCM2019} can be achieved without those restrictions. Much of the technical challenge therefore lies in proving a `converse theorem,' i.e., in showing a capacity upper bound that matches the expression in  \cite[Theorem 1]{JagannathanCM2019}. The key contributions in this paper can be summarised as follows:

\subsubsection{Upper bound on the capacity of additive quantum queue-channels} 
We first show that whenever the underlying quantum noise channel is \emph{additive}, we can obtain  a single-letter upper bound on the classical capacity of the corresponding quantum queue-channel. Our upper bound proof proceeds via the following key steps. The first step involves showing a certain conditional independence of $n$ consecutive channel uses, \emph{conditioned} on the sequence of qubit waiting times $(W_{1}, W_{2}, \ldots, W_{n})$. Specifically, we show that the $n$-qubit queue-channel factors into a \emph{tensor product} of single-use channels, for any given sequence of the waiting times. Next, we use a general capacity upper bound proved in~\cite[Lemma 5]{HN03_genCapacity} for non-stationary channels, which we then simplify using the conditional independence result from the first step.

\subsubsection{The Quantum Erasure Queue-Channel}
The general upper bound proved for any additive quantum queue-channel, along with the celebrated \emph{additivity} result of Holevo~\cite{holevo2004} for the quantum erasure channel immediately gives us a single-letter capacity upper bound. For the erasure queue-channel, we show that the upper bound is indeed the \emph{same} as the capacity expression derived in \cite[Theorem 1]{JagannathanCM2019} for the induced classical channel. That is, the capacity of the quantum erasure queue-channel does not increase by allowing for entangled channel uses by the encoder, and arbitrary measurements at the decoder. In other words, the  classical coding/decoding strategy in \cite{JagannathanCM2019} (proposed for the induced classical channel), is indeed sufficient to realise the classical capacity of the underlying  quantum erasure queue-channel. Furthermore, we show that the capacity remains the same, regardless of whether the arrival and departure times of the qubits (and hence their waiting times) is available at the decoder.  

In hindsight, this result may not be altogether surprising, considering that the classical capacity of the \emph{memoryless} quantum erasure channel is the same as the capacity of the classical erasure channel. In other words, a classical coding strategy is sufficient to realise the classical capacity of the (memoryless) quantum erasure channel --- see \cite{bennett97}.

\subsubsection{The Depolarising Queue-Channel} When the underlying noise model is a depolarising channel, we obtain the capacity expression for the queue channel, assuming that the decoder has access to the sequence of arrival and departure times (and hence the waiting times) of the qubits. The upper bound proof is technically similar to the erasure case, except that we appeal to the additivity result for the depolarising channel, due to King~\cite{king2003}. For the achievability part, we consider the corresponding induced classical channel (which turns out to be a binary symmetric queue-channel) and use a classical coding/decoding strategy to obtain an achievable rate which matches the upper bound. When the waiting times are not available at the decoder, the achievable rate that we obtain is strictly smaller than the upper bound enforced by additivity arguments --- in such a case, we can only identify the interval in which the capacity lies.

\subsection{Organisation} The remainder of this paper is organised as follows. Section~\ref{sec:model} details the system model and preliminaries. Section~\ref{sec:queueCapacity} presents the relevant channel capacity definitions and proceeds to present our main technical results. Section~\ref{sec:induced_class} deals with the induced classical channels and their capacities --- the achievability part of our capacity results also come from this section. Section~\ref{sec:concl} concludes the paper and proposes directions for future work.

\section{System Model and Preliminaries}
\label{sec:model}
The model we study is similar to the one considered in \cite{JagannathanCM2019}. Specifically, a source generates a classical bit stream, which is encoded into qubits. These qubits arrive at a continuous-time single-server queue according to a stationary point process of rate $\lambda.$  To be more explicit, the single-server queue is characterised by (i) A server that processes the qubits in the order in which they arrive, i.e., in a First Come First Served (FCFS) fashion\footnote{The FCFS assumption is not required for our results to hold, but it helps the exposition.}, and (ii) An `unlimited buffer' --- that is, there is no limit on the number of qubits that can queue up as they wait to be processed.  We assume that the service times for the qubits are independent and identically distributed (i.i.d.) random variables. The service time of the $j$th qubit is denoted by $S_j$ and has a cumulative distribution $F_S.$  The average service rate of each qubit is $\mu,$ i.e., $\EX_{F_S}[S] = 1/\mu.$ For stability of the queue, we assume $\lambda<\mu$. For ease of notation let us assume $\mu=1$. (Our results easily extend to general $\mu$). Let $A_j$ and $D_j$ be the arrival and the departure epochs of $j$th qubit, respectively and $W_j = D_j - A_j$ be the total time that $j$th qubit spends in the queue. Adopting a standard convention, we use upper case letters to denote random variables, and the corresponding lower case letters to denote the realized values of the random variable. For example, $W_j$ is the random variable that denotes the total time spent in the queue by the $j$th qubit, while $w_j$ denotes a specific realization of $W_j.$

The probability that the $j$th qubit undergoes erasure/error is modelled as an explicit function of its waiting time $w_j,$ as we specify in the next subsection. After getting processed at the  server, the qubits are measured and interpreted as classical bits (see Fig.~\ref{fig:diagram}).

\begin{remark}The main technical results in this paper hold as long as the queueing system described above is both \emph{stationary} and \emph{ergodic}. Stationarity essentially means that the queue eventually reaches a `steady-state' behaviour. In particular it implies that random variables of interest (such as waiting time of a qubit, queue length etc.) converge in distribution to a corresponding stationary distribution. Ergodicity means that the long-term \emph{time average} of waiting times, queue lengths, or functions thereof, converge (almost surely) to the corresponding \emph{ensemble average} (or expectation), taken over the stationary distribution. Most stable queuing models of interest (including the $\Mq/\GI/1$ queue and $\GI/\GI/1$ queue with non-arithmetic inter-arrival/service distributions) satisfy these two properties.

\textcolor{black}{We further remark that while we consider the queuing systems that are stationary, the queue-channel is non-stationary, because we start with an \emph{empty} queue, and send a \emph{finite} length codeword. Each symbol then undergoes waiting time dependent errors, which corresponds to a non-stationary channel.}
\end{remark}

\begin{figure}
  \centering
  \includegraphics[width=3.5in]{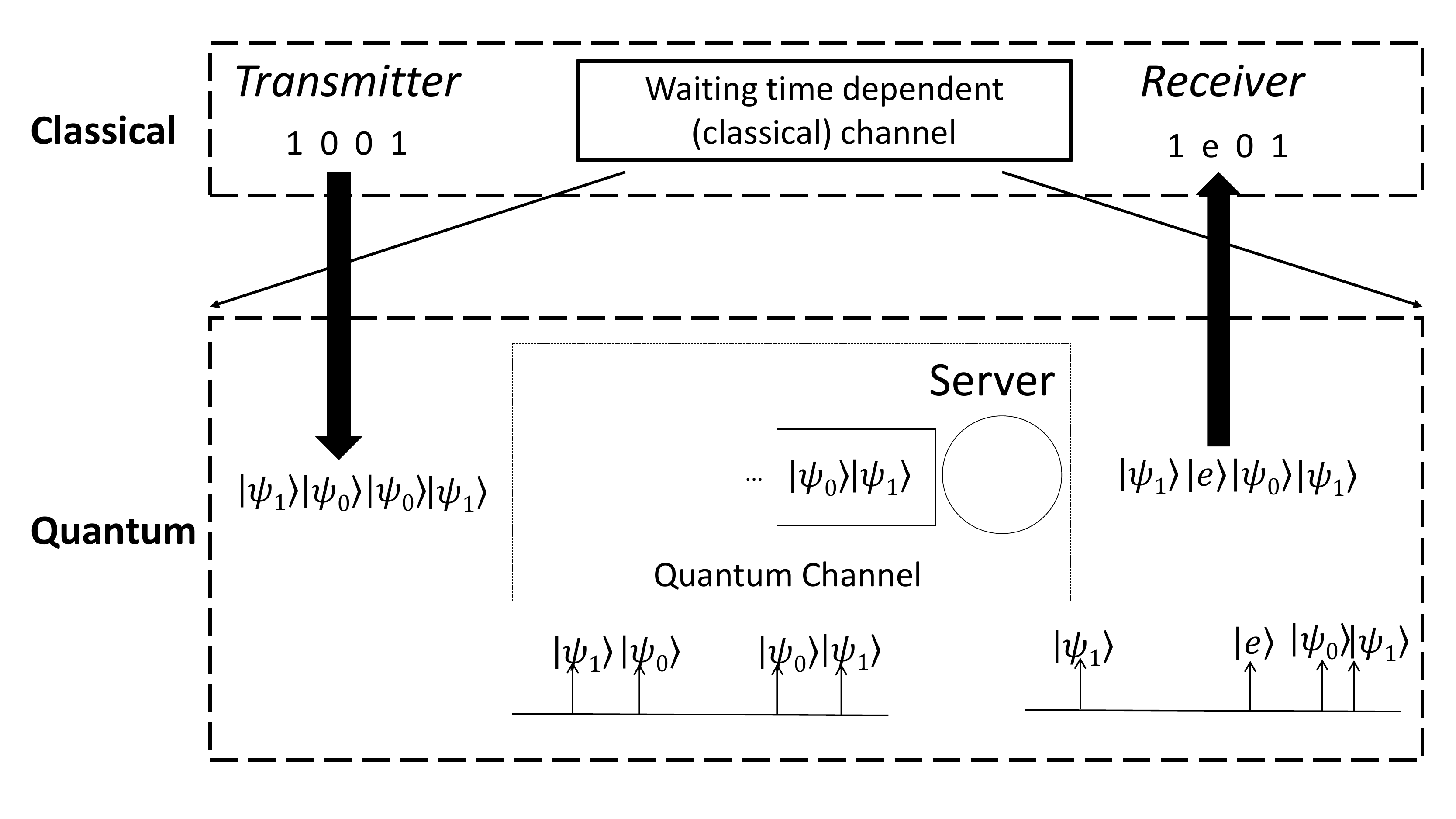}
  \caption{Schematic of the queue-channel depicting the case of quantum erasure.}
  \label{fig:diagram}
\end{figure}

\subsection{Noise Model}

As the qubits wait to be served, they undergo decoherence, leading to errors at the receiver. This decoherence is modeled mathematically as a completely positive trace preserving (CPTP) map on the the qubit states \cite{NCBook}. 

\textcolor{black}{Let the finite set $\mathcal X$ denote the input alphabet.} In general, a given input symbol $X_{j} \in \mathcal{X}$, is encoded as a positive trace-class operator (called the \emph{density operator}) $\rho_{j}$ on a Hilbert space $\mathcal{H}^{A}$ of dimension $\vert \mathcal{X}\vert$. The noise process is denoted as the map $\cN: \cS(\cH^{A}) \rightarrow \cS(\cH^{B})$, where $\cS(\cH^{A})$ ($\cS(\cH^{B})$) denotes the set of density operators on the \emph{input} Hilbert space $\cH^{A}$ (\emph{output} Hilbert space $\cH^{B}$). The map $\cN$ is often referred to as the (quantum) \emph{noise channel}. 

The probability that a given state $\rho_{j}$ undergoes decoherence is modelled as a function of the waiting time $W_{j},$ and the noise channel is accordingly parameterized in terms of the waiting time as $\cN_{W_{j}}$. The noisy output state after the action of the map $\cN_{W_{j}}$ is denoted as $\sigma_{j} \equiv \cN_{W_{j}}(\rho_{j})$. This noisy state is measured by the receiver by performing a general quantum measurement and decoded as the output symbol $Y_j \in \mathcal{Y}$. 

An $n$-length transmission over the above channel is denoted as follows. Inputs are drawn from the set $\cX^{(n)}$ of length $n$ symbols $\{ X^{n} = (X_{1}X_{2}\ldots X_n)\}$, and encoded into density operators $\rho_{X^{n}} \in \cS((\cH^{A})^{\otimes n})$ on an $n$-fold tensor product of the input Hilbert space $\cH^{A}$. The $n$-qubit channel is denoted $\cN^{(n)}_{W^n}$ and parameterized by the sequence of waiting times $W^{n} = (W_{1}, W_{2}, \ldots, W_{n})$. Note that the encoded state $\rho_{X^{n}}$ could be entangled across multiple channel uses. Furthermore, in general, the $n$-qubit queue-channel is not a stationary, memoryless channel and does not automatically factor into an $n$-fold tensor product of single qubit channels. We refer to the sequence of $n$-qubit channels $\vec{\cN}_{\vec{W}} = \{\cN^{(n)}_{W^{(n)}}\}_{n=1}^{\infty}$, which are parameterised by the corresponding waiting time sequences $\{W^{(n)}\}_{n=1}^{\infty}$, as a \emph{quantum queue-channel}, and characterize the classical capacity of this system (in bits/sec).

\textcolor{black}{In this paper, we  model the erasure/error probability of a qubit as an explicit function of its time spent in the queue. Thus, for the $j$th state $\rho_{j}$ with waiting time $w_j,$ we denote by $p(w_j)$ the probability its error/erasure, where $p:[0,\infty) \to [0,1]$ is typically increasing. We pay special attention to two important noise models:} \begin{itemize}
 \item[(i)]The  quantum erasure channel~\cite{wildeBook} which acts on the $j$th state $\rho_{j}$ with waiting time $w_j$ as follows: $\rho_{j}$ remains unaffected with probability $1 - p(w_{j})$, and is erased to a (fixed) erasure state $|e\rangle\langle e|$ with probability $p(w_{j}).$
\item[(ii)]The quantum depolarizing channel, which acts on a given qubit $\rho_{j}$ as, $\rho_{j} \rightarrow (1-p(w_{j}))\rho_{j} + p(w_{j})\frac{I}{2},$ where $I/2$ is the maximally mixed state~\cite{NCBook}.
\end{itemize}
 In several physical scenarios, the decoherence time of a single qubit maybe modelled as an exponential random variable. Thus, the  probability of a qubit erasure/error after waiting for a time $w$ is given by $p(w)=1-e^{-\kappa w}$, where $1/\kappa$ is a characteristic time constant of the physical system under consideration~\cite[Section 8.3]{NCBook}.

\section{Quantum Queue-Channel Capacity}
\label{sec:queueCapacity}
We are interested in defining and computing the information capacity of the quantum queue-channel. \textcolor{black}{We follow the conventions and the definitions from \cite{AnantharamV1996} that were adopted in \cite{ChatterjeeSV2017} for defining the capacity of queue-channels.}

\subsection{Definitions}
Let $M$ be the message transmitted from a set $\mathcal{M}$ and $\hat{M} \in \mathcal{M}$ be its estimate at the receiver. We now define an $(n, R, \epsilon, T)$ code for classical communication over a quantum-queue channel $\vec{\cN}_{\vec{W}}$. \textcolor{black}{Let $A^n, D^n$  denote the arrival epoch and departure epoch sequences, respectively.}
\begin{definition}
An $(n, R, \epsilon, T)$ quantum code is characterized by an encoding function $X^{n} = f(M)$, leading to an encoded $n$-qubit quantum state $\rho_{X^{n}}$ corresponding to message $M$, and a decoding function $$\hat{M} = g(\Lambda, \cN^{(n)}(\rho_{X^n}), A^n, D^n)$$ corresponding to a measurement $\Lambda$ at the receiver's end, {\color{black} where the cardinality of the message set $|\mathcal{M}| = 2^{nR}$, the expected total time for all the symbols of any codeword to reach the receiver is less than $T$, and the average probability of error of the code is less than $\epsilon$.

The measurement $\Lambda$ at the decoder is a positive operator valued measure (POVM) with operator elements $\{\Lambda_{\hat{M}}\}$ corresponding to message $M$.
The probability of error corresponding to message $M$ is \[ p_{e}(M) = \tr \left[ (I - \Lambda_{\hat{M}}) \cN^{(n)}_{W^{(n)}}(\rho_{X^n}) \right],\]
and the average probability of error is  given by $\frac{1}{|\mc{M}|}\sum_{M\in \mc{M}}p_e(M)$.}

{\color{black}As the average probability of error at the decoder for an $(n, R, \epsilon, T)$ quantum code is  less than $\epsilon$, the code is callled an $\epsilon$-achievable code with rate $\frac{R}{T}$.}

\end{definition}
\begin{definition}
 {\color{black}A rate $\bar{R}$ is said to be achievable if for any $\epsilon\in(0,1),$ there exists an $\epsilon$-achievable code with rate $\bar{R}$, or equivalently, if for any $\epsilon\in(0,1)$ there exists an $(n, {R}, \epsilon, T)$ code with $\bar{R}=\frac{R}{T}$.}
\end{definition}
\begin{definition}
\label{def:capacity}
The information capacity of the queue-channel is the supremum of all achievable rates for a given arrival and service process, and is denoted by $C$ bits per unit time.
\end{definition}
Note that the information capacity of the queue-channel depends on the arrival process, the service process, and the noise model. We assume that the receiver knows the realizations of the arrival and the departure times of each symbol, although we point out some results that hold even without this assumption.

\textcolor{black}{There are several physical scenarios where the arrival and departure times of each qubit is known, the simplest example being that of heralded single-photon sources which are ubiquitous in quantum communication networks (see~\cite{felinto2006} for example). The source typically produces a pair of entangled photons, one of which is used as a reference by the encoder/sender and the other photon is used in the actual communication protocol. The reference photon provides information about the arrival time, whereas the departure time can be obtained from the time-stamp on the receiver's detector.}

\textcolor{black}{Even in scenarios where this assumption may not be practical or realistic, conditioning on the waiting times offers a  useful conceptual handle towards tackling the problem. We also note that two key results in the paper---namely the general upper bound in Theorem~\ref{thm:queueCapacity_ub}, and the exact capacity of the erasure queue-channel in Theorem~\ref{thm:queueCapacityErasure}, hold irrespective of the receiver's knowledge of the waiting times.}


\subsection{Holevo Information and Additivity}

Before we proceed, we briefly review some well established concepts from quantum Shannon theory, which will be useful in analyzing the queue-channel capacity. An important measure that characterizes the classical capacity of quantum channels is the Holevo information of a quantum channel, also referred to as the Holevo capacity of a quantum channel~\cite{holevo1973}.
\begin{definition}[Holevo Information]
The Holevo Information of a quantum channel $\cN$ is defined as the entropy difference~\cite{wildeBook, watrous2018},
\begin{align}
\chi( \cN ) &:= \sup_{\{P_{x}, \rho_{x} \}} \chi (\{P_{x}, \cN(\rho_{x}) \}) \nonumber \\&= \sup_{\{P_{x}, \rho_{x} \}} H\left( \sum_{x}P_{x}\cN(\rho_{x}) \right) - \sum_{x}p_{x}H(\cN(\rho_{x})),
\end{align}
where the supremum is taken over all input ensembles $\{P_{x}, \rho_{x} \}$  and $H(\rho) = -\tr [\rho \log\rho]$ is the von Neumann entropy associated with a density operator $\rho$.
\end{definition}

The celebrated Holevo-Schumacher-Westmoreland theorem states that the classical capacity of a quantum channel, denoted as $C(\cN)$, is the \emph{regularized} Holevo information across independent channel uses~\cite{holevo1998,SW1997}. In other words,
\[C(\cN) = \lim_{n \rightarrow \infty}\frac{\chi(\cN^{\otimes n})}{n},\]
where $\cN^{\otimes n}$ is the product channel corresponding to $n$ independent channel uses. The Holevo information of the channel $\cN$ is said to be \emph{additive} if it satisfies $ \chi (\cN \otimes \cN) = 2\chi (\cN).$ It is easy to see that for quantum channels whose Holevo information is additive, the classical capacity of the channel is simply equal to the Holevo information. Additivity of the Holevo information further implies that the classical capacity of such channels is achievable via a \emph{classical} encoding and decoding strategy -- the optimal coding strategy does not require entangled inputs at the encoder or collective measurements at the decoder. Examples of channels for which the Holevo information is known to be additive include the quantum erasure channel~\cite{holevo2004} and the quantum depolarizing channel~\cite{king2003}. 

\begin{definition}[Additive quantum queue-channel] A quantum queue-channel $\vec{\cN}_{\vec{W}}$ is said to be additive if the Holevo information of the underlying single-use quantum channel $\cN$ is additive. Specifically, additivity of the Holevo information of the quantum channel $\cN$ implies
\[ \chi (\cN_{W_{1}} \otimes \cN_{W_{2}}) = \chi (\cN_{W_{1}}) + \chi (\cN_{W_{2}}) . \] 
\end{definition}

\subsection{Upper bound on the Capacity for Additive Quantum Queue-Channels} \label{sec:additive_ub}
We first invoke the capacity formula obtained in~\cite{HN03_genCapacity}, for the classical capacity of general quantum channels which are neither stationary nor memoryless. 
\begin{proposition}
\label{prop:expression}
The capacity of the quantum queue-channel (in bits/sec) described in Sec.~\ref{sec:model} is given by
\begin{align}
C = \lambda \sup_{\{\vec{P}, \vec{\rho}\}}\underline{\mathbf{I}} (\, \{\vec{P},\vec{\rho}\},\vec{\cN}_{\vec{W}} \, ),
\label{eq:cap_exp2}
\end{align}
where, $\underline{\mathbf{I}}( \, \{ \vec{P}, \vec{\rho} \, \}, \vec{\cN}_{\vec{W}} \, )$ is the quantum spectral inf-information rate originally defined in~\cite{HN03_genCapacity}. \textcolor{black}{We have stated this definition in Appendix~\ref{sec:qIbar} for completeness.} Here, $\vec{P}$ is the totality of sequences $\{P^{n}(X^{n})\}_{n=1}^{\infty}$ of probability distributions (with finite support) over input sequences $X^{n}$, and $\{\vec{\rho}\}$ denotes the sequences of states $\{\rho_{X^{n}}\}$ corresponding to the encoding $X^{n}\rightarrow \rho_{X^{n}}$. Finally,  $\vec{\cN}_{\vec{W}}$ denotes the sequence of channels $\{\cN^{(n)}_{W^{(n)}}\}_{n=1}^{\infty}$, which are parameterised by the corresponding waiting time sequences $\{W^{(n)}\}_{n=1}^{\infty}$.
\end{proposition}
\begin{figure*}[!t]
	\begin{eqnarray}
	& & \ \ \cN^{(n)}_{W^{(n)}} (\rho_{12\ldots n}) \nonumber \\
	&=& \sum_{k_{1}, k_{2}, \ldots, k_{n}} q_{k_{1}k_{2}\ldots k_{n}}(W_{1},W_{2}, \ldots, W_{n}) \, \qN_{k_{1}}\otimes \qN_{k_{2}}\ldots \otimes \qN_{k_n} \left(\rho_{12\ldots n}\right) \qN^{\dagger}_{k_{1}}\otimes \qN^{\dagger}_{k_{2}}\ldots \otimes \qN^{\dagger}_{k_{n}} \nonumber \\
	&=& \sum_{k_{1}, k_{2}, \ldots, k_{n}} q_{k_{1}}(W_{1})q_{k_{2}}(W_{2})\ldots q_{k_{n}}(W_{n}) \qN_{k_{1}}\otimes \qN_{k_{2}}\ldots \otimes \qN_{k_{n}} \left(\rho_{12\ldots n}\right) \qN^{\dagger}_{k_{1}}\otimes \qN^{\dagger}_{k_{2}}\ldots \otimes \qN^{\dagger}_{k_{n}} \nonumber \\
	&=& \left(\cN_{W_{1}}\otimes \cN_{W_{2}}\ldots \otimes \cN_{W_{n}} \right) (\rho_{12\ldots n}) . \label{eq:tensorproduct}
	\end{eqnarray}
	\hrulefill
\end{figure*}
We are now ready to state and prove a general upper bound for the capacity of additive quantum queue-channels. 
\begin{theorem}\label{thm:queueCapacity_ub}
For an additive quantum queue-channel $\vec{\cN}_{\vec{W}}$, the capacity is bounded as, 
\[ C \leq \lambda~\EX_{\pi}\left[ \chi(\cN_{W}) \right]\ {\rm bits/sec.}, \] 
irrespective of the receiver's knowledge of the arrival and the departure times. Here $\chi(\cN_{W})$ denotes the Holevo information of the single-use quantum channel corresponding to waiting time $W$ and $\pi$ is the stationary distribution of the waiting time in the queue.
\end{theorem}
\emph{Proof Outline:}
Obtaining the queue-channel capacity of a quantum channel poses certain technical challenges, since the error probabilities are correlated across different channel uses. In other words, the probability that the $i^{\rm th}$ qubit gets affected by an error is a function of its waiting time in the queue, which in turn depends on the waiting time of the previous $(i-1)^{\rm th}$ qubit and so on. Furthermore, the channel is non-stationary. However, for the queue-channel model considered here, the $n$-qubit queue-channel does factor into a \emph{tensor product} of single-use channels, \emph{conditioned} on the sequence of waiting times $(W_{1}, W_{2}, \ldots, W_{n})$. This is formally shown in Lemma~\ref{lem:cond_ind} below. 

In order to obtain an \emph{upper bound on capacity}, we proceed via the following key steps. First, we invoke an upper bound proved in~\cite[Lemma 5]{HN03_genCapacity}, to bound the capacity as the limit inferior of the \emph{Holevo information} of a sequence of quantum channels. Next, we use the tensor product from of the channel obtained in Lemma~\ref{lem:cond_ind} in conjunction with the fact that the Holevo information of the channel $\cN$ is additive. Finally, we invoke the ergodicity of the queue to obtain a single-letter expression for an upper bound on the queue-channel capacity.

\begin{lemma}[Conditional independence]\label{lem:cond_ind}
The $n$-qubit quantum queue-channel $\cN^{(n)}_{W^n}$ factors into a tensor product of single-use channels, conditioned on the waiting times  $(W_{1}, W_{2}, \ldots, W_{n})$.
\end{lemma}
\begin{proof}
Consider a sequence of $n$ qubits transmitted via a quantum queue-channel with associated waiting times $W_{1}, W_{2}, \ldots, W_{n}$. The $n$-qubit channel maybe described via $n$-fold tensor-product operators of the form \textcolor{black}{$\{\qN_{1}\otimes \qN_{2}\otimes \ldots \otimes \qN_{n}\}$}, where each single-qubit operator $\qN_{i}$ is one of a finite set of \emph{noise} operators $\{E_{1}, E_{2}, \ldots, E_{L}\}$, which characterize the quantum channel $\cN$. Note that the operators $E_{i} \in \mathbb{M}_{2 \times 2}$ which belong to the space of $2 \times 2$ complex matrices are called the \emph{Kraus} operators associated with the quantum channel $\cN$. 

There are $L^{n}$ tensor-product operators of the form \textcolor{black}{$\{\qN_{1}\otimes \qN_{2}\otimes \ldots \otimes \qN_{n}\}$}, occurring with probabilities $q_{k_{1}k_{2}\ldots k_{n}} (W_{1}, W_{2}, \ldots, W_{n})$, with the indices $k_{i} \in \{ E_{1},E_{2},\ldots, E_{L}\}$, depending on which noise operator each $\qN_{i}$ corresponds to. This $n$-fold channel is a non-iid, correlated quantum channel in general, since the joint distribution $q_{k_{1}k_{2}\ldots k_{n}} (W_{1}, W_{2}, \ldots, W_{n})$ does not factor into a product of the individual error probabilities for each qubit. However, conditioned on the the waiting time sequence $W^{(n)} = (W_{1}, W_{2}, \ldots, W_{n})$, the joint distribution does factor as $ q_{k_{1}k_{2}\ldots k_{n}} (W_{1}, W_{2}, \ldots, W_{n}) = \Pi_{i} q_{k_{i}}(W_{i}).$
Therefore, conditioned on the waiting times $(W_{1}, W_{2}, \ldots, W_{n})$ we may represent the action of the $n$-qubit channel on any $n$-qubit state $\rho_{12\ldots n}$ as shown in \eqref{eq:tensorproduct}. In other words, conditioned on the waiting time sequence $W^{n}$, the $n$-qubit channel factors into an $n$-fold tensor product of the form $\cN_{W_{1}}\otimes \cN_{W_{2}}\otimes \ldots \otimes \cN_{W_{n}}$, as desired.
\end{proof}

We are now ready to prove the upper bound on the capacity of an additive quantum queue-channel. We assume that the sequence of waiting times $\vec{W}$ is available at the receiver.

\begin{proof} (Theorem~\ref{thm:queueCapacity_ub})
We start with an upperbound on the quantum inf-information rate proved in~\cite[Lemma 5]{HN03_genCapacity}:
\[ \underline{\mathbf{I}}( \, \{\vec{P},\vec{\rho} \, \},  \vec{\cN}_{\vec{W}}) \leq \liminf_{n\rightarrow \infty} \frac{1}{n} \chi( \, \{P^{(n)}, \rho_{X^{n}}\} ,\cN^{(n)}_{W^{(n)}} \, ) ,\]
where, $\chi( \{P^{(n)}, \rho_{X^{n}}\}, \cN^{(n)}_{W^{n}})$ is the Holevo information of the ensemble $\{ P^{(n)}(X^{n}), \cN^{(n)}_{W^{(n)}}(\rho_{X^{n}}) \}$. Consider now the Holevo information $\chi(\cN^{(n)}_{W^{(n)}})$ of the $n$-qubit quantum queue-channel, for a given sequence of waiting times $W^{(n)}$. Lemma~\ref{lem:cond_ind} implies that,
\[ \chi(\cN^{(n)}_{W^{n}}) = \chi (\cN_{W_{1}} \otimes\cN_{W_{2}}\otimes \ldots \otimes \cN_{W_{n}} ) .\]
Furthermore, the fact that the queue-channel is additive implies for any $n$, 
\[ \chi(\cN^{(n)}_{W^{n}}) = \sum_{i=1}^{n}\chi \left( \, \cN_{W_{i}} \, \right) .\]
Rewriting this in terms of the Holevo information of the encoding ensemble for the $n$-qubit channel, we get,
{\small
\begin{align}
 \sup_{ \{P^{n}(X^{n}), \rho_{X^{n}}\} } \chi( \, \{ P^{n}_{X^{n}}, \cN^{(n)}_{W^{(n)}}(\rho_{X^{n}}) \} \, )&  \nonumber \\= \sum_{i=1}^{n} \sup_{ \{ P_{i}(X_{i}), \rho_{i}\}} &\chi( \, \{ P(X_{i}), \cN_{W_{i}}(\rho_{i}) \} \, ) . \label{eq:additivity}
\end{align}
}
Combining the above sequence of steps, we thus get the following upper bound on the capacity of an additive quantum queue-channel:
\begin{eqnarray}
C &\stackrel{(a)}{=}& \lambda \sup_{\{\vec{P}, \vec{\rho}\}}\underline{\mathbf{I}} (\, \{\vec{P},\vec{\rho}\},\vec{\cN}_{\vec{W}} \, ) \nonumber \\
&\stackrel{(b)}{\leq} & \lambda \sup_{ \{\vec{P}, \vec{\rho}\} } \liminf_{n\rightarrow \infty} \frac{1}{n} \chi( \, \{P^{(n)}, \rho_{X^{n}}\} ,\cN^{(n)}_{W^{(n)}} \, ) \nonumber \\
&\stackrel{(c)}{\leq}& \lambda \liminf_{n\rightarrow \infty} \frac{1}{n} \sup_{\{P^{(n)}, \rho_{X^{n}}\}} \chi( \, \{P^{(n)}, \rho_{X^{n}}\} ,\cN^{(n)}_{W^{(n)}} \, ) \nonumber \\
&\stackrel{(d)}{=}& \lambda \liminf_{n\rightarrow \infty} \frac{1}{n} \sum_{i=1}^{n} \sup_{\{P(X_{i}), \rho_{i}\}}\chi( \, \{ P(X_{i}), \cN_{W_{i}}(\rho_{i}) \} \nonumber \\
&=& \lambda \liminf_{n\rightarrow \infty} \frac{1}{n} \sum_{i=1}^{n} \chi \left( \, \cN_{W_{i}} \, \right) \nonumber \\
&\stackrel{(e)}{=}& \lambda~\EX_{\pi}\left[ \chi(\cN_{W})\right]~{\rm a.s}. \nonumber
\end{eqnarray}
Here, $(a)$ is simply the definition of the queue-channel capacity as stated in Eq.~\eqref{eq:cap_exp2}, and $(b)$ is the upper bound from~\cite[Lemma 5]{HN03_genCapacity}. This upper bound uses the quantum Neyman-Pearson Lemma and the monotonicity of the quantum relative entropy, in addition to invoking arguments from quantum hypothesis testing. The inequality $(c)$ follows from elementary real analysis. The equality in $(d)$ follows from the conditional independence of the $n$-use channel and the additivity of the quantum queue-channel (see Eq.~\eqref{eq:additivity}), and $(e)$ follows from the ergodicity of the queue.
\end{proof}

We note that our proof of the upper bound on classical queue-channel capacity assumes knowledge of the sequence of waiting times $(W_{1}, W_{2}, \ldots, W_{n})$ at the receiver. This automatically implies an upper bound for the scenario where the receiver does not have knowledge of the waiting times.

\subsection{Erasure Queue-Channels}
\label{sec:ECquantum}

Erasure channels are ubiquitous in classical as well as quantum information theory.  Our model of an erasure queue-channel captures a quantum information system where qubits decohere with time into erased (or non-informative) quantum states.

A single-use quantum erasure channel, for a qubit with waiting time $W_{1}=w_1$ is characterized by a pair of noise operators (or \emph{Kraus} operators), namely the operator $E: \rho_{1} \rightarrow |e\rangle\langle e|$ which maps any input density operator $\rho_{1}$ to a fixed erasure state $|e\rangle\langle e|$ , and the identity operator $I : \rho_{1}\rightarrow \rho_{1}$. If the input state has a waiting time $W_{1}=w_1$, the erasure and identity operations occur with probabilities $q_{E}(w_1) = p(w_1)$ and $q_{I}(w_1) = 1 - p(w_1)$ respectively. Thus, the final state after the action of the erasure queue-channel for a given input state $\rho_{1}$, with waiting time $W_{1}=w_1$ is
\[ \cN_{W_{1}}(\rho_{1}) = q_{E}(w_1) \, E \rho_{1} E^{\dagger} + q_{I}(w_1) \, \rho_{1}  . \]

We now evaluate the upper bound proved in Theorem~\ref{thm:queueCapacity_ub} for the case of the quantum erasure queue-channel and show that is achievable using a classical coding strategy. This leads us to the following result on the classical capacity of the quantum erasure queue-channel.

\begin{theorem}
\label{thm:queueCapacityErasure}
For the erasure queue-channel defined above, the capacity is given by $C = \lambda~\EX_{\pi}\left[1-p(W)\right]$ bits/sec, irrespective of the receiver's knowledge of the arrival and the departure times, where $\pi$ is the stationary distribution of the waiting time in the queue.
\end{theorem}
\emph{Proof Outline:} The upper bound simply follows from the upper bound proved in Theorem~\ref{thm:queueCapacity_ub} above, since the Holevo capacity of the quantum erasure channel is additive~\cite{holevo2004}. The \emph{achievability} proof follows by fixing a classical encoding and decoding strategy and showing that the capacity of the induced classical channel does indeed coincide with the upper bound. 

\begin{proposition}[Upper bound on Capacity]\label{prop:upper_bound}
The capacity of the quantum erasure queue-channel satisfies $C \leq \lambda~\EX_{\pi}\left[1-p(W)\right]$  bits/sec. 
\end{proposition}
\begin{proof}
In order to obtain an \emph{upper bound on capacity}, we invoke the celebrated additivity result of Holevo~\cite{holevo2004} for the quantum erasure channel. Next, we use the fact that the Holevo information of a single use erasure channel corresponding to waiting time $W_{i}$  is given by~\cite{bennett97},
\begin{equation}
 \chi (\cN_{W_{i}}) \equiv  \sup_{\{P(x), \rho_{x}\}}\chi( \, \{P_{x}, \cN_{W_{i}}(\rho_{x}) \} \, )  = 1 - p(W_{i}) . \label{eq:erasure_cap}
 \end{equation}
This, along with the upper bound proved in Theorem~\ref{thm:queueCapacity_ub} above, results in the desired upper bound. 
\end{proof}


\begin{proposition}[Lower bound on Capacity (Achievability)]

\label{prop:lower_bound}
The capacity of the quantum erasure queue-channel satisfies
$C \geq \lambda~\EX_{\pi}\left[1-p(W)\right]$ bits/sec.
\end{proposition}
\begin{proof}
We prove the lower bound by producing a particular encoding/decoding strategy that achieves the said capacity expression. In particular, we employ a classical strategy in which classical bits 0 and 1 are encoded into two fixed orthogonal states (say $|\psi_0\rangle$ and $|\psi_1\rangle$), and the decoder also measures in a fixed basis. The input codewords are unentangled across multiple channel uses and the decoder simply performs  product measurements. In this setting, the qubits essentially behave as classical bits, and the quantum erasure channel  simulates the \emph{induced classical} channel. The capacity of the corresponding induced classical channel is equal to $\lambda~\EX_{\pi}\left[1-p(W)\right],$ as shown in Theorem~\ref{thm:queueCapacityErasureInd} in Sec.~\ref{sec:EC} below.
\end{proof}

We remark that the above capacity result does not depend on the specific functional form of $p(\cdot).$ Further, the capacity result holds for any stationary and ergodic queue --- i.e., it does not assume any specific queueing model. If we assume the functional form $p(W)=1-\exp(-\kappa W)$ for the erasure probability, the following  corollary is immediate.
\begin{corollary}
\label{thm:queueCapacityErasurekappa}
When the decoherence time of each qubit is exponentially distributed, i.e., $p(W)=1-\exp(-\kappa W),$ the erasure queue-channel capacity is given by $\lambda \EX_{\pi}\left[e^{-\kappa W}\right]$ bits/sec., where $\pi$ is the stationary distribution of the waiting time in the queue.
\end{corollary}   

\subsection{$\Mq/\Mq/1$ and $\Mq/\GI/1$ examples}
We remark that the capacity expression $\lambda\EX_{\pi}\left[e^{-\kappa W}\right]$ is simply $\lambda$ times the Laplace transform of the stationary waiting time $W,$ evaluated at $\kappa,$ which is the rate of decoherence. We now derive closed form expressions for the capacity for the $\Mq/\Mq/1$ and $\Mq/\GI/1$ queues. An $\Mq/\GI/1$ queue is a well-studied model of a single server queue in which the arrival process is a Poisson process, the service times are independent of the arrival process, and generally distributed (but i.i.d across qubits). See \cite[Chapter 5]{Kleinrock1975_I} for a detailed treatment of the $\Mq/\GI/1$ queue. The $\Mq/\Mq/1$ queue is is a special case of an $\Mq/\GI/1$ queue, where the service times are exponentially distributed.

Using Pollaczek-Khinchin formula \textcolor{black}{(see \cite[Eq. (5.105)]{Kleinrock1975_I})} for an FCFS $\Mq/\GI/1$ queue, we can obtain a closed-form expression for the capacity and the optimal arrival rate. 
\begin{theorem}
\label{thm:bestArrMG1}
For an FCFS $\Mq/\GI/1$ erasure queue-channel (with $p(W)=1-\exp(-\kappa W)$), 
\begin{itemize}
\item[(i)] the capacity is given by $\frac{\lambda(1-\lambda)}{1-\alpha\lambda}$ bits/sec, and
\item[(ii)] the capacity is maximised at $$\lambda_{\Mq/\GI/1} = \frac{1}{\alpha} \left(1-\sqrt{1-\alpha}\right)=\frac{1}{1+\sqrt{1-\alpha}},$$\end{itemize}
where $\alpha=\frac{1-\tilde{F}_S(\kappa)}{\kappa}$, and $\tilde{F}_S(u)=\int \exp(-u x) dF_S(x)$ is the Laplace transform of the service time distribution.
\end{theorem}
\begin{IEEEproof}[Proof]
First, note that for the particular form of $p(\cdot)$ considered here, the capacity expression in Theorem~\ref{thm:queueCapacityErasure} simplifies to
$\lambda \EX[\exp(-\kappa W)].$

Next, using the  Pollaczek-Khinchin formula for the $\Mq/\GI/1$ queue (with $\mu=1$), we write
\begin{align}
& 
\ \EX[\exp(-\kappa W)]  \nonumber  = \frac{(1-\lambda)\kappa}{\kappa - \lambda (1-\tilde{F}_S(\kappa))} \nonumber  = \frac{1-\lambda}{1 - \alpha \lambda},
\end{align}
where $\alpha = \frac{(1-\tilde{F}_S(\kappa))}{\kappa}$. Thus, the capacity is $\frac{\lambda (1-\lambda)}{1 - \alpha \lambda}$ bits/sec, and the capacity maximising arrival rate
is $$\arg\max_{\lambda \in [0,1)} \frac{\lambda (1-\lambda)}{1 - \alpha \lambda}\mbox{.}$$

The objective function in the above optimization problem is
 concave in $\lambda$. This implies that the value of $\lambda$ that maximises the capacity  is the one at which  the derivative of the capacity with respect to $\lambda$ is zero. Taking the derivative, we obtain a quadratic function in $\lambda$ which when equated to zero yields two solutions for $\lambda$:
$\frac{1}{\alpha} \pm \frac{\sqrt{1-\alpha}}{\alpha}.$
The only valid solution for which $\lambda\in[0,1)$ is given by $\frac{1}{\alpha} - \frac{\sqrt{1-\alpha}}{\alpha}$.
\end{IEEEproof}

\begin{figure}
\centering
  \includegraphics[scale=0.7]{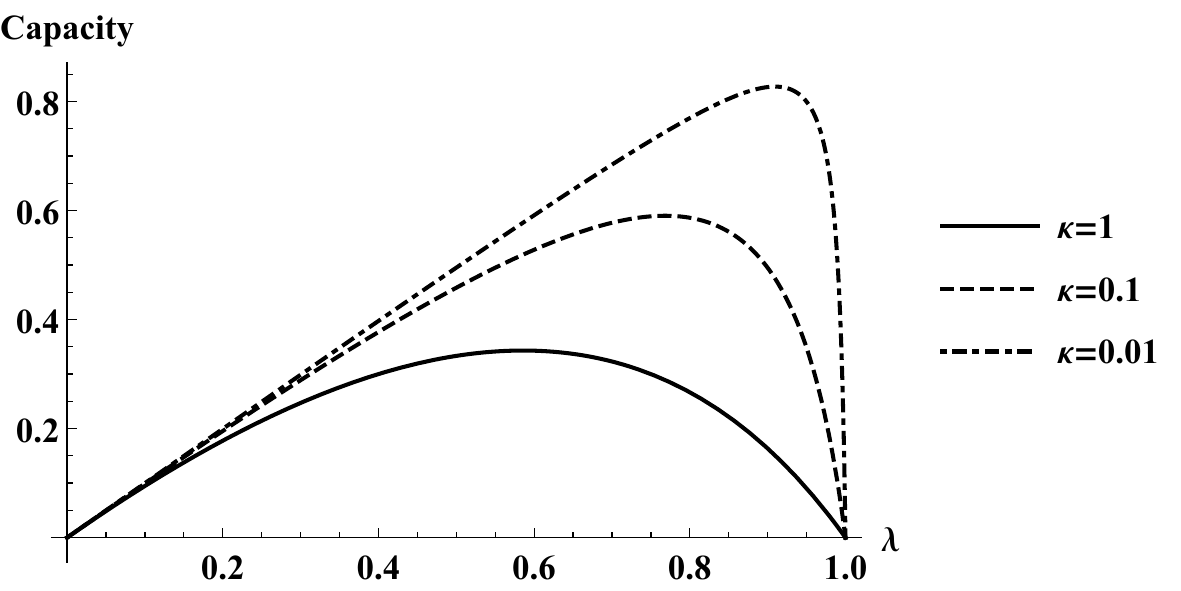}
  \caption{The capacity of the $\Mq/\Mq/1$ erasure queue-channel (in bits/sec) plotted as a function of the arrival rate $\lambda$ for different values of the decoherence parameter $\kappa.$}
  \label{fig:plot}
\end{figure}

This result offers interesting insights into the relation between the information capacity and the characteristic time-constant of the quantum states. Fig.~\ref{fig:plot} plots capacity versus arrival rate for an $\Mq/\Mq/1$ queue of unit service rate. We note that $\kappa=0.01$ corresponds to an average coherence time which is two orders of magnitude longer than the service time --- a setting  reminiscent of superconducting qubits~\cite{wendin2017quantum}. We also notice from the shape of the capacity curve for $\kappa=0.01$ that there is a drastic drop in the capacity, if the system is operated beyond the optimal arrival rate $\lambda_{\Mq/\Mq/1}.$ This is due to the drastic increase in delay induced decoherence as the arrival rate of qubits approaches the server capacity.

Theorem~\ref{thm:bestArrMG1} characterizes an optimal $\lambda$ for given arrival and service distributions. One can also
ask after the best service distribution for a given arrival process and a fixed server rate. This question is of interest
in designing the server characteristics like gate operations \cite{gatefidelity16}, or photon detectors for quantum communications. The following theorem is useful in such scenarios.

\begin{theorem}
\label{thm:bestServMG1}
For an erasure queue-channel with $p(W)=1-\exp(-\kappa W)$ and Poisson arrivals at a given rate $\lambda,$  the capacity is maximised  by the service time distribution $F_S(x)=\mathbf{1}(x\ge 1)$, i.e., a deterministic service time maximises capacity, among all
service distributions with unit mean and $F_S(0)=0$.
\end{theorem}
\begin{IEEEproof}[Proof]
As derived in the proof of Theorem \ref{thm:bestArrMG1}, the capacity is
\begin{align}
\frac{\lambda (1-\lambda)\kappa}{\kappa - \lambda (1-\tilde{F}_S(\kappa))} \nonumber = \frac{\frac{(1-\lambda)\kappa}{\lambda}}{\frac{\kappa-\lambda}{\lambda} + \tilde{F}_S(\kappa)}\mbox{.}
\nonumber
\end{align}
Thus, for any given $\lambda$, among all service distribution with unit mean, the capacity is maximised 
by that service distribution for which $\tilde{F}_S(\kappa)$ is minimised.
For any service random variable $S$, by Jensen's inequality, we have
$\tilde{F}_S(\kappa)=\EX[\exp(-\kappa S)]\ge \exp(-\kappa\EX[S]).$
Therefore, $\tilde{F}_S(\kappa)$ is minimised by $S=\EX[S]$, i.e., a deterministic service time.
\end{IEEEproof}

Theorems~\ref{thm:bestArrMG1} and \ref{thm:bestServMG1} hold for an  $\Mq/\GI/1$ queue when $p(W)=1-e^{-\kappa W}.$ For the  $\Mq/\Mq/1$ case, we can derive a capacity expression for any functional form of the erasure probability $p(\cdot).$
\begin{theorem}
\label{thm:bestArrMM1GenPsi}
For the $\Mq/\Mq/1$ erasure queue-channel, the capacity expression for any functional form of the erasure probability $p(\cdot)$ is given by $$\lambda \left(1 - \frac{1-\lambda}{\lambda}\tilde{p}\left(\frac{1-\lambda}{\lambda}\right)\right)\ {\rm bits/sec.},$$ where for any $u>0$, $\tilde{p}(u):=\int \exp(-u x) p(x) dx$ is the Laplace transform of $p(\cdot)$. Furthermore, the capacity maximizing arrival rate is 
 \begin{align}
& 1 - \arg\min_{u \in (0,1)} u~\left(1 + \tilde{p}\left(\frac{u}{1-u}\right)\right).
\label{eq:bestArrMM1GenPsi} 
 \end{align}
 
\end{theorem}

\begin{IEEEproof}[Proof]
This proof uses the exponential waiting time distribution of $\Mq/\Mq/1$ queue to relate
the capacity to Laplace transform of $p(\cdot)$. It is known that the waiting time in $\Mq/\Mq/1$ is distributed as $\exp\left(\frac{1-\lambda}{\lambda}\right)$ for $\mu=1$. Thus, 
\begin{align}
& \ \EX[p(W)] \nonumber 
= \int_{0}^{\infty} p(w) \frac{1-\lambda}{\lambda} \exp\left(\frac{1-\lambda}{\lambda}w\right)
dw \nonumber \\
& = \frac{1-\lambda}{\lambda}\tilde{p}\left(\frac{1-\lambda}{\lambda}\right). \nonumber
\end{align}

Thus, the capacity is given by
$\lambda \left(1 - \frac{1-\lambda}{\lambda}\tilde{p}\left(\frac{1-\lambda}{\lambda}\right)\right)$ bits/sec.

Therefore, the capacity maximising arrival rate is the one that maximises this expression:
\begin{align}
& \ \ \arg\max_{\lambda \in (0,1)} \lambda \left(1 - \frac{1-\lambda}{\lambda}\tilde{p}\left(\frac{1-\lambda}{\lambda}\right)\right) \nonumber \\ 
& \iff \arg\max_{\lambda \in (0,1)} \left(\lambda  - (1-\lambda)\tilde{p}\left(\frac{1-\lambda}{\lambda}\right)\right) \nonumber \\
& \iff 1 - \arg\max_{u \in (0,1)} \left(1 - u - u \tilde{p}\left(\frac{u}{1-u}\right)\right)\nonumber \\
& \iff 1 - \arg\min_{u \in (0,1)} u~\left(1 + \tilde{p}\left(\frac{u}{1-u}\right)\right).
\nonumber
 \end{align}
\end{IEEEproof}
\subsection{Quantum Depolarizing Queue-channel}\label{sec:depolarizing}

Our techniques can be easily extended to study another important additive quantum channel, namely the depolarizing channel. Suppose that the $j$th qubit spends $W_j=w_j$ seconds in the queue. The quantum depolarizing channel acts on a given qubit $\rho_{j}$ as, $\rho_{j} \rightarrow (1-p(w_{j}))\rho_{j} + p(w_{j})\frac{I}{2},$ where $I/2$ is the maximally mixed state~\cite{NCBook}. Additivity of the Holevo information of the depolarizing channel was shown by King~\cite{king2003}, leading to a single-letter formula for its classical capacity, originally noted in~\cite{bennett97}. Here, we evaluate the classical capacity of the quantum depolarizing queue-channel as stated below.
\begin{theorem}
The classical capacity of a quantum depolarizing queue-channel with arrival rate $\lambda$ is, 
\[ C = \lambda~\EX_\pi\left[1-h\left(\frac{p(W)}{2}\right)\right] \; {\rm bits/sec},\]
provided the receiver knows the arrival and the departure times of the qubits. Here, $h(.)$ denotes the binary entropy function, and the expectation is taken over the stationary distribution of the waiting time $W$ in the queue.
\end{theorem}
\begin{proof}
Since the depolarising channel satisfies additivity \cite{king2003}, the upper bound in Theorem~\ref{thm:queueCapacity_ub} holds. Furthermore, the Holevo information for a single-qubit depolarising channel is evaluated to be~\cite{bennett97}
\[  \chi(\cN_{W})=1-h\left(\frac{p(W)}{2}\right). \]
To prove achievability, we first note that the induced classical channel corresponding to the quantum depolarizing channel is a binary symmetric channel with crossover probability $\phi(w)=p(w)/2$ (see ~\cite{bennett97}). Then, we invoke the capacity result proved in Theorem~\ref{thm:queueCapacityBSC} in Sec.~\ref{sec:bsc} for the induced classical channel.  
\end{proof}
\begin{remark}
We note that for the depolarising queue channel, when the arrival and departure times are known to the receiver, we get matching upper and lower bounds from Theorems~\ref{thm:queueCapacity_ub} and \ref{thm:queueCapacityBSC}, thus characterising the capacity exactly. However, when the arrival and departure times are not known to the receiver, there is a `Jensen gap' between the upper and lower bounds in Theorems~\ref{thm:queueCapacity_ub} and \ref{thm:queueCapacityBSC} respectively. Thus, when the waiting times are not known to the receiver, we are unable to precisely characterise the capacity --- we  know only that it lies 
in between $\lambda\left(1 - h\left(\EX_{\pi}\left[\frac{p(W)}{2}\right]\right)\right)$ and $\lambda~\EX_\pi\left[1-h\left(\frac{p(W)}{2}\right)\right]$ bits/sec. 
\end{remark}

\begin{remark}
We remark that the upper bound and the capacity results which have been stated and proved for qubit queue-channels, do extend readily to the case of qu$d$it ($d$-dimensional) queue-channels as well. The upper bound on the queue-channel capacity stated in Theorem~\ref{thm:queueCapacity_ub} is of course independent of the dimension of the states and channels involved; rather it depends only on the additivity of the Holevo information for the channel. However, the capacity expressions do change with the dimensionality of the system. Specifically, the classical capacity of a qu$d$it quantum erasure queue-channel is $C_{d} = \lambda\log_{2} d~\EX_{\pi}\left[1-p(W)\right]$ bits/sec. and the classical capacity of a qu$d$it depolarizing queue-channel (when the waiting times are known to the receiver) is 

{\small \begin{align*}C_{d} = \lambda~\EX_\pi\left[\log_{2}d + \left(1-p+\frac{p}{d}\right)\log_{2}\left(1-p+\frac{p}{d}\right) +\right.\\\left. (d-1)\frac{p}{d}\log_{2}\frac{p}{d}\right] {\rm bits/sec.} \end{align*}}
\end{remark}

\section{Capacity of the induced classical queue-channel}\label{sec:induced_class}

In this section we define and characterize the capacity of the induced classical channel corresponding to a general quantum queue-channel. Formally, the $j$th symbol $X_j \in \mathcal{X}=\{0,1\}$\footnote{Results extend to any finite alphabet.}, is now encoded as one of a set of orthogonal states $\{|\psi_{X_{j}}\rangle\}$ belonging to a Hilbert space $\mathcal{H}$ of dimension $2$. The noisy output state $|\tilde{\psi}_{j}\rangle$ is measured by the receiver in some fixed basis, and decoded as the output symbol $Y_j \in \mathcal{Y}$. This measurement induces a conditional probability distribution $\PR(Y_j|X_j, W_j)$. The \emph{induced classical channel} is a sequence of conditional distributions from ${\mathcal{X}}^n$ to ${\mathcal{Y}}^n$:
\begin{equation}\PR(Y^{n}|X^{n}, W^{n})=\prod_{j=1}^{n} \PR(Y_j|X_j, W_j) \mbox{ for }n \in \{1, 2, \ldots\}.\label{eq:productform}\end{equation}
As argued in Lemma \ref{lem:cond_ind}, $Y_j$ is conditionally independent of other random variables, given $X_j$ and $W_j,$ which leads to the product form above.
 Throughout, $Z^{n} = (Z_1, Z_2, \ldots, Z_n)$ denotes an $n$-dimensional vector
and $\mbf{Z}=(Z_1, Z_2, \ldots, Z_n, \ldots)$ denotes an infinite sequence of random variables.

We are interested in defining and computing the information capacity of this induced classical queue-channel. As mentioned earlier, we restrict  ourselves to using a fixed set of orthogonal states to encode the classical symbols at the sender's side, and measuring in some fixed basis at the receiver's end. For this reason, the capacity of the induced classical channel is not the same as the \emph{classical capacity of the quantum channel} resulting from the underlying decoherence model.  

The motivation behind the study of the induced classical queue-channel is twofold. Firstly, the achievability result (i.e., the capacity lower bound) for this channel is applicable to the classical capacity of the quantum queue-channel studied in the previous section. Indeed, as we show below, when the arrival and departure times of the symbols are known at the receiver, the capacities of the  induced classical erasure and binary symmetric queue-channels are the same as that of the  classical capacities of the corresponding quantum queue-channels.
Secondly, the induced queue-channel has other interesting applications as well. It closely models delay-sensitive communication systems where data packets become useless after a deadline~\cite{JagannathanCM2019,ChatterjeeJM2018Long}.

\subsection{Definitions}
Let $M$ be the message transmitted from a set $\mathcal{M}$ and $\hat{M} \in \mathcal{M}$ be its estimate at the receiver.
\begin{definition}
An $(n, \tilde{R}, \epsilon, T)$ code consists of the encoding function $X^n=f(M)$ and the decoding function $\hat{M} = g(X^n, A^n, D^n)$, where the cardinality of the message set $|\mathcal{M}| = 2^{n\tilde{R}}$, the expected total time for all the symbols of any codeword to reach to the receiver is less than $T$, and the average probability of error at the decoder is less than $\epsilon$.

If the average probability of error of a code is less than $\epsilon$ the code is called an $\epsilon$-achievable code.
\end{definition}
\begin{definition}
For any $0 < \epsilon < 1$, if there exists an $\epsilon$-achievable code $(n, \tilde{R}, \epsilon, T)$, the rate $R = \frac{\tilde{R}}{T}$ is said to be achievable.
\end{definition}
\begin{definition}
\label{def:capacityInd}
The information capacity of the induced classical queue-channel is the supremum of all achievable rates for a given arrival and service process and is denoted by $C_{Ind}$ bits per unit time.
\end{definition}

We assume that the transmitter knows the arrival process statistics, but not the realizations before it
does the encoding.  However, depending on the scenario, the receiver may or may not know the realization of the arrival and the departure time of each symbol.

The following result is a consequence of the general channel capacity expression in~\cite{VerduH1994} and is useful in deriving the main results on single letter capacity expression.

\begin{proposition}
\label{prop:expressionInd}
The capacity of the queue-channel (in bits/sec) described above is given by
\begin{align}
C_{Ind} = \lambda \sup_{\PR(\mathbf{X})}\underline{\mathbf{I}} (\mbf{X}; \mbf{Y} | \mbf{W}),
\label{eq:cap_expInd}
\end{align}
when the receiver knows the arrival and the departure time of each symbol. On the other hand, when the receiver does not have that information, the capacity is, 
\begin{align}
C_{Ind} = \lambda \sup_{\PR(\mathbf{X})}\underline{\mathbf{I}} (\mbf{X}; \mbf{Y})\mbox{.} 
\label{eq:cap_exp3}
\end{align}
Here, $\underline{\mathbf{I}}$ is the usual notation for classical \emph{inf-information rate} \cite{VerduH1994}. Specifically, $\underline{\mathbf{I}} (\mbf{X}; \mbf{Y})$ is defined as the limit inferior in probability of the \emph{information density} sequence $\left\{\frac1n\log\frac{\PR(Y^n|X^n)}{\PR(Y^n)}\right\}_{n=1}^\infty.$
\end{proposition}

\begin{IEEEproof}
The case where arrival and departure times of the qubits are not known follows directly from~\cite{VerduH1994} using (limiting) stationarity and ergodicity of the arrival and the departure processes of the queue. Note that a string of $n$ qubits see the joint channel 
$$\PR({Y}^n|{X}^n)=\int_{W^n} \prod_{j=1}^n \PR(Y_j|X_j, W_j) d\PR(W^n),$$
and the number of qubits coming out of the queue per unit time (asymptotically) is 
$\lambda$. Combining these two facts with standard information spectrum results we get the desired capacity result.

When the arrival and the departures times of the qubits are known at the receiver, the channel 
behaves like $X^n \rightarrow (Y^n, A^n, D^n)$. In that case it follows from~\cite{VerduH1994} that 
the capacity of this channel is
$$\lambda \sup_{\PR(\mathbf{X})}\underline{\mathbf{I}} (\mbf{X}; (\mbf{Y},\mbf{A},\mbf{D})).$$
Next, note that $W_j=D_j-A_j$ for all $j,$ and since no information is encoded in timings, $\mbf{X}$ is independent of $(\mbf{A}, \mbf{D})$. By definition,  $\underline{\mathbf{I}} (\mbf{X}; (\mbf{Y},\mbf{A},\mbf{D}))$ is the limit inferior in probability of the following sequence:
\begin{align}
& \ \ \frac{1}{n} \log \frac{\PR(Y^n, A^n, D^n|X^n)}{\PR(Y^n, A^n, D^n)} \nonumber \\
& =  \frac{1}{n} \log \frac{\PR(A^n, D^n|X^n) \PR(Y^n|X^n, A^n, D^n)}{\PR(A^n, D^n)\PR(Y^n|A^n, D^n)} \nonumber \\
& = \frac{1}{n} \log \frac{\PR(Y^n|X^n, A^n, D^n)}{\PR(Y^n|A^n, D^n)} \nonumber \\
& = \frac{1}{n} \log \frac{\PR(Y^n|X^n, W^n)}{\PR(Y^n|W^n)}\label{eq:IXYW}
\end{align}
The step before the last is due to independence of $X^n$ and $(A^n, D^n)$. Also, as per the channel model we have the Markov property $(A^n, D^n) \rightarrow W^n \rightarrow Y^n$, which leads to the final expression. The liminf in probability of the quantity in \eqref{eq:IXYW}, is by definition, equal to $\underline{\mathbf{I}} (\mbf{X}; \mbf{Y} | \mbf{W}).$
\end{IEEEproof}

\subsection{Erasure Queue-Channels}
\label{sec:EC}

Erasure channels are ubiquitous in classical as well as quantum information theory. We consider a classical erasure queue-channel~\cite{bennett97} induced by the quantum erasure queue-channel where the $j$th state $|\psi_{X_{j}}\rangle $ remains unaffected with probability $1 - p(W_{j})$, and is erased to a state $|e\rangle$ with probability $p(W_{j})$, where $p:[0,\infty) \to [0,1]$ is typically increasing. Such a model also captures the communication scenarios where information packets become useless (erased) after a deadline. For such an erasure channel, a single letter expression for capacity can be obtained.

\begin{theorem}
\label{thm:queueCapacityErasureInd}
For the classical erasure queue-channel defined above, the capacity (in bits/sec) is given by $\lambda~\EX_{\pi}\left[1-p(W)\right]$, irrespective of the receiver's knowledge of the arrival and the departure times of symbols.
\end{theorem}
\begin{remark}
For a general queue-channel, the capacity could depend on the receiver's knowledge (or lack thereof) of the arrival and the departure times of symbols. However, for the erasure case,  the received symbol is either correct or is erased, but is never wrong. This makes the knowledge of the arrival and departure times irrelevant, and the capacity remains the same in both cases. 
\end{remark}
\begin{IEEEproof}[Proof of Theorem~\ref{thm:queueCapacityErasureInd}]
We first prove the result for the case when the receiver knows the arrival and the departure times. 

The proof uses  an upper-bound on $\underline{\mathbf{I}}(\mbf{X};\mbf{Y}|\mbf{W})$ in terms of two conditional sup-entropy rates rate and shows that $\EX_{\pi}\left[1-p(W)\right]$ is an upper-bound on $\underline{\mathbf{I}}(\mbf{X};\mbf{Y}|\mbf{W})$. On the other hand, using a similar lower-bound on $\underline{\mathbf{I}}(\mbf{X};\mbf{Y}|\mbf{W})$ we show that for a choice of distribution of $\{X_n\}$ (namely, i.i.d. uniform),  $\underline{\mathbf{I}}(\mbf{X};\mbf{Y}|\mbf{W})$ is no smaller than $\EX_{\pi}\left[1-p(W)\right]$. More precisely, we use the following two bounds, which follow from the properties of limit superior and limit inferior~\cite[Chap.~3]{Han2003}:
\begin{eqnarray}
		\underline{\mathbf{I}} (\mbf{X}; \mbf{Y} | \mbf{W}) &\le& \overline{\mbf{H}}(\mbf{Y}|\mbf{W}) - \overline{\mbf{H}}(\mbf{Y}|\mbf{X},\mbf{W}), \label{eq:limsup_I} \\
\underline{\mathbf{I}} (\mbf{X}; \mbf{Y} | \mbf{W}) &\ge& \underline{\mbf{H}}(\mbf{Y}|\mbf{W}) - \overline{\mbf{H}}(\mbf{Y}|\mbf{X},\mbf{W}), \label{eq:liminf_I}
\end{eqnarray}
where $\overline{\mbf{H}}(\mbf{Y}|\mbf{W})$ and $\overline{\mbf{H}}(\mbf{Y}|\mbf{X},\mbf{W})$ are respectively the lim-sup in probability of the sequences $\left\{\frac{1}{n} \log \frac{1}{\PR(Y^n|W^n)}\right\}_{n=1}^\infty$ and $\left\{\frac{1}{n} \log \frac{1}{\PR(Y^n|X^n,W^n)}\right\}_{n=1}^\infty$. Similarly, $\underline{\mbf{H}}(\mbf{Y}|\mbf{W})$ is the lim-inf in probability of $\left\{\frac{1}{n} \log \frac{1}{\PR(Y^n|W^n)}\right\}_{n=1}^\infty$. 

We now state three lemmas, which lead to the capacity result. 
\begin{lemma}
\label{lem:HUbarYgivenXandWErasure}
For the induced classical erasure queue-channel,
$$\overline{\mbf{H}}(\mbf{Y}|\mbf{X},\mbf{W})=\EX_\pi[h(p(W))],$$
irrespective of the choice of $\PR_{\mbf{X}}$.
\end{lemma}

\begin{lemma}
\label{lem:HUbarYgivenWErasure}
For the induced classical erasure queue-channel,
$$\overline{\mbf{H}}(\mbf{Y}|\mbf{W})\le 1+\EX_\pi[h(p(W))]-\EX_\pi[p(W)], $$
irrespective of the choice of $\PR_{\mbf{X}}$.
\end{lemma}

\begin{lemma}
\label{lem:HUbarYgivenWErasureAchieve}
For uniform and i.i.d $\{X_i\}$
$$\underline{\mbf{H}}(\mbf{Y}|\mbf{W})=1+\EX_\pi[h(p(W))]-\EX_\pi[p(W)].$$
\end{lemma}
The proofs of Lemmas~\ref{lem:HUbarYgivenXandWErasure}, \ref{lem:HUbarYgivenWErasure}, and \ref{lem:HUbarYgivenWErasureAchieve} appear in Appendix~\ref{sec:proofs}. Combining Lemmas~\ref{lem:HUbarYgivenXandWErasure} and \ref{lem:HUbarYgivenWErasure} we get a $\PR_{\mbf{X}}$-independent upper-bound on $\underline{\mathbf{I}}(\mbf{X};\mbf{Y}|\mbf{W}):$ $$\underline{\mathbf{I}} (\mbf{X}; \mbf{Y} | \mbf{W})\leq 1-\EX_\pi[p(W)].$$ A matching lower bound on $\underline{\mathbf{I}}(\mbf{X};\mbf{Y}|\mbf{W})$ is readily obtained using  Lemmas~\ref{lem:HUbarYgivenXandWErasure} and \ref{lem:HUbarYgivenWErasureAchieve}. Thus, for the erasure queue-channel,
$$\sup_{\PR_X} \underline{\mathbf{I}}(\mbf{X};\mbf{Y}|\mbf{W})=1-\EX_\pi[p(W)].$$
We obtain the capacity of this channel by multiplying the above expression by $\lambda$ (see Proposition~\ref{prop:expressionInd}). 

When the arrival and departure times are \emph{not} available  at the receiver, one way to prove the capacity result  is to analyze $\underline{\mathbf{I}}(\mbf{X};\mbf{Y})$. However, it is easier to use a dominance argument for the probability of errors in the two scenarios: arrival and departure times known, versus unknown, at the receiver. Consider any encoding and decoding scheme which attains capacity when the arrival and departure times are known at the receiver. Note that under this encoding scheme the rate is close to the capacity and the probability of error is tending to $0$. By the optimality of MAP (maximum aposterior probability) decoding, if we replace the decoding scheme by a MAP decoder the probability error cannot increase. For the case where the receiver does not know arrival and departure times, we can use the \emph{same} encoding scheme, followed by a MAP decoding.  For an erasure channel, the MAP decoding depends only on the prior probabilities of the codewords that match the received codeword at the non-erased symbols. Therefore, the knowledge of $\{p(W_i)\}$ (and hence that of arrival and departure times) has no effect on the MAP decoding. Thus, it follows that the same rate and the probability of error are also achievable in the absence of arrival and departure times.
\end{IEEEproof}

Note that the above theorem, particularly the achievability result, completes the proof of Proposition~\ref{prop:lower_bound}.
Furthermore, since the capacity of the induced classical erasure queue-channel is same as the classical capacity of the quantum erasure queue-channel, the
results in Theorem~\ref{thm:bestArrMG1}, \ref{thm:bestServMG1}, and \ref{thm:bestArrMM1GenPsi} are applicable here as well. 

\subsection{ Binary Symmetric Queue-channels}\label{sec:bsc}

In this section, we state our results for a binary symmetric queue-channel. This classical channel is induced by the depolarizing quantum queue-channel. 
 The probability of a bit flip is captured by a function $\phi:[0,\infty) \to [0,0.5]$, $\phi(W_i)$, where
 $\PR(Y_i\neq X_i|X_i, W_i) = \phi(W_i)$. Note that the $p(W)$ for the depolarizing channel and $\phi(W)$ are related as $\phi(W)=\frac{p(W)}{2}$.

\begin{theorem}
\label{thm:queueCapacityBSC}
For a binary symmetric queue-channel, the capacity is 
$\lambda\left(1 - \EX_{\pi}\left[h(\phi(W))\right]\right)$ bits/sec 
when the receiver knows the arrival and departure times of the symbols. The capacity is lower bounded by
$\lambda\left(1 - h(\EX_{\pi}\left[\phi(W)\right])\right)$ when the receiver does not know the arrival and departure times
of the symbols.
\end{theorem}
\begin{IEEEproof}

\noindent {\em Case I: arrival and departures times are known at the receiver}

\noindent These steps are similar to the proof of the capacity result for the erasure queue-channel.  Once again, we use the inequalities in \eqref{eq:limsup_I} and \eqref{eq:liminf_I}
respectively, for the upper and the lower bounds. For proving the upper-bound, we argue along the same lines as Lemmas~\ref{lem:HUbarYgivenXandWErasure} and \ref{lem:HUbarYgivenWErasure} to obtain
\begin{align}
 \ \ \overline{\mbf{H}}(\mbf{Y}|\mbf{W}) &\le 1, \mbox{ and} \nonumber \\
 \ \ \overline{\mbf{H}}(\mbf{Y}|\mbf{X},\mbf{W}) &= \EX_{\pi}\left[h(\phi(W))\right]. \label{eq:BSCHUbarEquality}
\end{align}
The last equality uses the following simple insight: the  channel probability distribution for binary symmetric queue-channel can be written as 
$$\PR(Y^n|X^n,W^n)=\PR(N^n|W^n),$$
where $Y_i = X_i + N_i~(\mbox{modulo } 2)$ and $N_i \sim \mbox{Bernoulli}(\phi(W_i))$. The rest follows using 
conditional independence of $\{N_i\}$ given $\{W_i\},$ and the ergodicity of the process $\{W_i\}$. This completes the proof of the upper-bound.

For deriving the lower bound, we fix the input distribution to be i.i.d. Bernoulli$(0.5)$ independent of $\mbf{W}$. An elementary calculation then shows that $\mbf{Y}$ is also i.i.d. Bernoulli$(0.5)$ and is  independent of $\mbf{W}$. Thus,
$$\underline{\mbf{H}}(\mbf{Y}|\mbf{W})=1.$$ 
Using \eqref{eq:liminf_I} and \eqref{eq:BSCHUbarEquality} we obtain the following lower bound:
$$\underline{\mathbf{I}} (\mbf{X}; \mbf{Y} | \mbf{W}) \ge 1 - h(\EX_{\pi}\left[\phi(W)\right]).$$

\noindent {\em Case II: arrival and departures times are not known at the receiver}

For obtaining the lower-bound in this case we use the following inequality from \cite{Han2003}.
$$\underline{\mathbf{I}} (\mbf{X}; \mbf{Y}) \ge \underline{\mbf{H}}(\mbf{Y}) - \overline{\mbf{H}}(\mbf{Y}|\mbf{X}).$$
Again, for  uniform i.i.d. $\mbf{X}$, we have
$\underline{\mbf{H}}(\mbf{Y})=1.$

For obtaining a lower bound, we obtain an upper-bound on
$\overline{\mbf{H}}(\mbf{Y}|\mbf{X})$ by upper-bounding the limit 
\begin{align}
& \ \  \lim_{n\to\infty}\frac{1}{n} \log\left(\frac{1}{\PR(N^n)}\right) \label{eq:NergodicLimit}
\end{align}
in an almost sure sense.

Since $\{W_i\}$ is stationary and ergodic, so is $\{N_i\}$. To bound the limit in \eqref{eq:NergodicLimit}, we use successive Markov approximations of this ergodic process, where the Markov approximation approaches the `true' process, as the order or the memory of the Markov process tends to infinity. (A celebrated application of this successive Markov approximation approach appears in proving the asymptotic optimality of the LZ78 universal source coding algorithm --- see \cite[Section~13.5.2]{CoverT2006}).

To be precise, we first consider the case where $N_i$ is a $k$-th order stationary and ergodic Markov process and obtain a bound on the above limit,  which does not depend on $k$. 
Let $\{\tilde{N}_t\}$ be a $k$-th order Markov process with a transition kernel 
\begin{align}
& \prod_{j=1}^n \PR(\tilde{N}_j=u_j|\tilde{N}_{j-1}=u_{j-1},  \ldots \tilde{N}_{(j-k)_+}=u_{(j-k)_+}) \nonumber \\
& =\prod_{j=1}^n \PR(N_j=u_j|N_{j-1}=u_{j-1}, \ldots N_{(j-k)_+}=u_{(j-k)_+}),
\nonumber 
\end{align}
and $\PR(\tilde{N}_1=u_1)=\PR(N_1=u_1)$, where $(j-k)_+=\max(j-k,1)$. As claimed in the proof of \cite[Lemma~13.5.4]{CoverT2006}, \eqref{eq:NergodicLimit} is equal to 
$\lim_{k\to \infty} \beta_k$, where $\beta_k$ is the almost sure limit of 
$\frac{1}{n} \log\left(\frac{1}{\PR(\tilde{N}^n)}\right)$ for the $k$th order Markov process.

Note that
$\frac{1}{n} \log\left(\frac{1}{\PR(\tilde{N}^n)}\right)$ is given by
$$
\frac{1}{n} \log\left(\frac{1}{\PR(\tilde{N}^k)}\right) - \frac{1}{n} \sum_{t=k+1}^n \log \PR(\tilde{N}_t|\tilde{N}_{t-1},  \ldots, \tilde{N}_{t-k+1}).$$
By ergodicity, the above expression has the following almost sure limit:
$$\EX_{\tilde{N}_{k+1},\tilde{N}_{k}, \ldots, \tilde{N}_{1}}\left[- \log \PR(\tilde{N}_{k+1}|\tilde{N}_{1}, \tilde{N}_{2}, \cdots, \tilde{N}_{k})\right].$$
The expectation above is over the joint distribution of $(\tilde{N}_{1}, \tilde{N}_{2}, \cdots, \tilde{N}_{k}, \tilde{N}_{k+1}),$ and can be  recognized readily as the conditional entropy of $\tilde{N}_{k+1},$ given $\tilde{N}_{1}, \tilde{N}_{2}, \cdots, \tilde{N}_{k}$. This in turn implies that
\begin{align}
& \ \ \EX_{\tilde{N}_{k+1},\tilde{N}_{k}, \ldots, \tilde{N}_{1}}\left[- \log \PR(\tilde{N}_{k+1}|\tilde{N}_{1}, \tilde{N}_{2}, \cdots, \tilde{N}_{k})\right] \nonumber \\
& = \EX_{\tilde{N}_{k}, \ldots, \tilde{N}_{1}}\left[\EX_{\tilde{N}_{k+1}}\left[- \log \PR(\tilde{N}_{k+1}|\tilde{N}_{1}, \tilde{N}_{2}, \cdots, \tilde{N}_{k})\right]\right] \label{eq:iterExpe} \\
& = \EX_{\tilde{N}_{k}, \ldots, \tilde{N}_{1}}\left[h(\PR(\tilde{N}_{k+1}|\tilde{N}_{1}, \tilde{N}_{2}, \cdots, \tilde{N}_{k}))\right] \nonumber \\
& \le h(\EX_{\tilde{N}_{k}, \ldots, \tilde{N}_{1}}\left[\PR(\tilde{N}_{k+1}|\tilde{N}_{1}, \tilde{N}_{2}, \cdots, \tilde{N}_{k})\right]) \label{eq:concavity}\\
& = h(\PR(\tilde{N}_{k})) = h(\EX_{\pi}\left[\phi(W)\right]). \label{eq:stationarity} 
\end{align}
Here, \eqref{eq:iterExpe}, \eqref{eq:concavity} and \eqref{eq:stationarity} follow from iterated expectation, Jensen's inequality due to concavity of $h(\cdot)$, and stationarity of $\tilde{N}_{k}$ which is {Bernoulli}$(\EX_{\pi}\left[\phi(W)\right])$, respectively.
Thus, $h(\EX_{\pi}\left[\phi(W)\right])$ is an upper-bound on $\EX_{\tilde{N}_{k+1},\tilde{N}_{k}, \ldots, \tilde{N}_{1}}\left[- \log \PR(\tilde{N}_{k+1}|\tilde{N}_{1}, \tilde{N}_{2}, \cdots, \tilde{N}_{k})\right]$ for all $k$.
Hence, \eqref{eq:NergodicLimit} is upper-bounded by $h(\EX_{\pi}\left[\phi(W)\right])$.
\end{IEEEproof}
When the receiver does not know the arrival and departure times, $\lambda\left(1 - \EX_{\pi}\left[h(\phi(W))\right]\right)$ is still  an upper-bound on the capacity. Thus, there remains a `Jensen's gap' between the upper and the lower bound in that case.

We remark that the above capacity result for the binary symmetric queue-channel is a special case of  the capacity result for a class of classical queue-channels called random bijective queue-channels --- see \cite{ChatterjeeJM2018Long}. 
\section{Concluding Remarks and Future Work}
\label{sec:concl}
In this paper, we considered quantum queue-channels as a framework to study the sequential processing of qubits, and derived fundamental capacity limits for sequential quantum information processing. For the class of additive queue-channels, we derived a single-letter upper bound for the capacity in terms of the stationary waiting time in the queue. We then showed that this capacity upper bound is achievable for two important noise models, namely the erasure channel and the depolarising channel. We also showed that a classical coding/decoding strategy is capacity achieving for these channel.
\begin{figure*}[t!]
	\begin{equation}
	\underline{\mathbf{I}}( \, \{ \vec{P}, \vec{\rho} \, \}, \vec{\cN}_{\vec{W}} \, ) = 
	\sup \left\lbrace a\in\mathbb R^+\left\vert \lim_{n\rightarrow \infty} \sum_{X^{n}\in \cX^{(n)}} P^{n}(X^{n}) \tr \left[ \cN^{(n)}_{W^{(n)}}(\rho_{X^{n}}) \left\lbrace \Gamma_{\{P^{n}(X^{n}), \rho_{X^{n}}\}}(a)   > 0 \right.\right\rbrace \right] = 1 \right\rbrace .\label{eq:quantum_I}
	\end{equation}
	\hrulefill
\end{figure*}
\begin{figure*}[t!]
	\begin{align}
	& \ \ \ - \log \PR(Y^n|W^n) = - \log \PR(\{Y_i: i \in J^{\mc{E}}\}, \{Y_i: i \not\in J^{\mc{E}}\} |W^n) \nonumber \\
	& =  - \log \PR(\{Y_i: i \in J^{\mc{E}}\}|\{W_i: i \in J^{\mc{E}}\})  - \log \PR(\{Y_i: i \not\in J^{\mc{E}}\}|\{W_i: i \not\in J^{\mc{E}}\})
	\label{eq:erasure2}\\
	& = - \sum_{i \in J^{\mc{E}}} \log \PR(\mc{E}|W_i) - \log \PR(\{Y_i: i \not\in J^{\mc{E}}\}|\{W_i: i \not\in J^{\mc{E}}\}) \label{eq:erasure3} \\
	& = \left. - \sum_{i \in J^{\mc{E}}} \log \PR(\mc{E}|W_i) -  \log \PR(\{Y_i\neq \mc{E}:i \not\in J^{\mc{E}}\},\{Y_i: i \not\in J^{\mc{E}}\}|\{W_i: i \not\in J^{\mc{E}}\})\right.  \label{eq:erasure4} \\
	& = \left. - \sum_{i \in J^{\mc{E}}} \log p(W_i)  - \log \PR(\{Y_i\neq \mc{E}:i \not\in J^{\mc{E}}\}|\{W_i: i \not\in J^{\mc{E}}\}) \right.  \left. - \ \log \PR(\{Y_i: i \not\in J^{\mc{E}}\}|\{Y_i\neq \mc{E}:i \not\in J^{\mc{E}}\},\{W_i: i \not\in J^{\mc{E}}\})\right. \nonumber \\
	& = \left. - \sum_{i \in J^{\mc{E}}} \log p(W_i) - \sum_{i \not\in J^{\mc{E}}} 
	\log (1-p(W_i)) \right.  \left.    -  \ \log \PR(\{Y_i: i \not\in J^{\mc{E}}\}|\{Y_i\neq \mc{E}:i \not\in J^{\mc{E}}\},\{W_i: i \not\in J^{\mc{E}}\})\right. \nonumber \\
	&= - \sum_{i=1}^n \left[\mbf{1}(Y_i=\mc{E}) \log(p(W_i)) + \mbf{1}(Y_i\neq \mc{E}) \log(1-p(W_i))\right]    -  \log \PR(\{Y_i: i \not\in J^{\mc{E}}\}|\{Y_i\neq \mc{E}, W_i: i \not\in J^{\mc{E}}\})\mbox{.} \label{eq:erasure5}
	\end{align}
	\hrulefill
\end{figure*}

There is ample scope for further work along several directions. First, we can study the queue-channel capacity for other noise channels whose Holevo information is known to be additive. This includes the class of entanglement breaking channels~\cite{shor_EB} as well as the class of Hadamard channels~\cite{wildeBook}. The upper bound on the classical capacity of the queue-channel proved here will indeed hold in all of these cases. However, whether the upper bound is achievable or not will depend on the structure of the induced classical channel. As seen in the case of the depolarizing queue-channel, the upper bound may or may not be tight depending on whether the receiver's knowledge of the waiting times of the qubits. Whether our capacity results can be extended to the larger class of additive queue-channels is an interesting avenue for further investigations. In this context, the upper bounds proved in~\cite{wang2019converse} and the corresponding achievability results for finite-length codes might be useful.
An immediate application of our work would be in identifying optimal rates for reliable communication across quantum networks. We can also quantitatively evaluate the impact of using quantum codes with finite block lengths to protect qubits from errors. Employing a code would enhance robustness to errors, but would also increase the waiting time due to the increased number of qubits to be processed. It would be interesting to characterise this tradeoff, and identify the regimes where using coded qubits would be beneficial or otherwise.

More generally, the qubits-through-queues paradigm studied here may be used to analyse various performance aspects of quantum circuits, such as, quantifying the best achievable processing rate for asymptotically perfect accuracy. As we enter an era of quantum networks and noisy intermediate-scale quantum technologies~\cite{preskill2018}, our work begins to quantitatively address the impact of decoherence on the performance limits of quantum information
processing systems.

\appendix

\subsection{Quantum Inf-Information Rate}\label{sec:qIbar}
\textcolor{black}{Consider the quantum queue-channel $\vec{\cN}_{\vec{W}}$ comprising a sequence of channels $\{\cN^{(n)}_{W^{(n)}}\}_{n=1}^{\infty}$, which are parameterised by the corresponding waiting time sequences $\{W^{(n)}\}_{n=1}^{\infty}$. Let $\vec{P}$  denote the totality of sequences $\{P^{n}(X^{n})\}_{n=1}^{\infty}$ of probability distributions (with finite support) over input sequences $X^{n}$, and $\{\vec{\rho}\}$ denote the sequences of states $\{\rho_{X^{n}}\}$ corresponding to the encoding $X^{n}\rightarrow \rho_{X^{n}}$. For any $a \in \mathbb{R}^{+}$ and $n$, we first define the operator
	\begin{align*}\Gamma_{\{P^{n}(X^{n}), \rho_{X^{n}}\}}(a) &=\\ \cN^{(n)}_{W^{(n)}}&(\rho_{X^{n}})- e^{an}\sum_{X^{n}\in \cX^{(n)}}P^{n}(X^{n})\cN^{(n)}_{W^{(n)}}(\rho_{X^{n}}) .\end{align*}
	Let $\{  \Gamma_{\{P^{n}(X^{n}), \rho_{X^{n}}\}}(a)  > 0\}$ denote the projector onto the positive eigenspace of the operator $\Gamma_{\{P^{n}(X^{n}), \rho_{X^{n}}\}}(a) $. 
	\begin{definition}
		The quantum inf-information rate~\cite{HN03_genCapacity} $\underline{\mathbf{I}}( \, \{ \vec{P}, \vec{\rho} \, \}, \vec{\cN}_{\vec{W}} \, )$  is defined as in \eqref{eq:quantum_I}.
	\end{definition}
	This is the quantum analogue of the classical inf-information rate originally defined in~\cite{VerduH1994, Han2003}. The central result of~\cite{HN03_genCapacity} is to show that the classical capacity of the channel sequence $\vec{\cN}_{\vec{W}}$ is given by
	\[C = \sup_{\{\vec{P}, \vec{\rho}\}}\underline{\mathbf{I}} (\{\vec{P},\vec{\rho}\},\vec{\cN}_{\vec{W}}). \]
}

\subsection{Proofs of Lemmas~\ref{lem:HUbarYgivenXandWErasure}, \ref{lem:HUbarYgivenWErasure} and~\ref{lem:HUbarYgivenWErasureAchieve}}\label{sec:proofs}

\begin{IEEEproof}[Proof of Lemma~\ref{lem:HUbarYgivenXandWErasure}]
	We are interested in computing $\overline{\mbf{H}}(\mbf{Y}|\mbf{X},\mbf{W}),$ which is the limsup in probability of the sequence $\left\{\frac{1}{n} \log \frac{1}{\PR(Y^n|X^n,W^n)}\right\}_{n=1}^\infty.$ We start with the product form \eqref{eq:productform} for the channel, and also note that $\PR(Y_i|X_i, W_i)=p(W_i)$ if $Y_i$ is an erasure, else, it is $1-p(W_i)$. Combining these observations we obtain
	
	{\small \begin{align}
		& \ \ \log \frac{1}{\PR(Y^n|X^n,W^n)} = \sum_{i=1}^n \log\frac{1}{\PR(Y_i|X_i,W_i)} \nonumber \\
		& = \left[\sum_{i \in J^{\mc{E}}}\log\frac{1}{\PR(Y_i|X_i,W_i)} + \sum_{i \in [n]\setminus J^{\mc{E}}}\log\frac{1}{\PR(Y_i|X_i,W_i)}\right]\nonumber\\
		& = - \sum_{i=1}^n \left[\mbf{1}(Y_i=\mc{E}) \log(p(W_i)) + \mbf{1}(Y_i\neq \mc{E}) \log(1-p(W_i))\right]\mbox{,}\nonumber 
		\end{align}}
	where $\mc{E}$ represents erasure, and  $J^{\mc{E}}$ denotes the set of indices for which $Y_i=\mc{E}$. By ergodicity of the queue, the limit of the final expression divided by $n$ exists almost surely as a finite constant, and is equal to $\EX_\pi[h(p(W))].$ 
	
	Finally, we note that the lim-sup in probability and the lim-inf in probability of a sequence are both equal to the almost sure limit of that sequence, whenever the latter exists.
	Hence, the above almost sure limit is  the value of $\overline{\mbf{H}}(\mbf{Y}|\mbf{X},\mbf{W})$. Note that for an erasure queue-channel this limit does not depend on the distribution of $X^n.$
\end{IEEEproof}

\begin{IEEEproof}[Proof of Lemma~\ref{lem:HUbarYgivenWErasure}] We are interested in computing $\overline{\mbf{H}}(\mbf{Y}|\mbf{W}),$ which is the limsup in probability of the sequence $\left\{\frac{1}{n} \log \frac{1}{\PR(Y^n|W^n)}\right\}_{n=1}^\infty.$ 
	Following standard conditional probability arguments, we get the series of equalities 
	\eqref{eq:erasure2}--\eqref{eq:erasure5}.
	
	The equalities in \eqref{eq:erasure2} and \eqref{eq:erasure3} follow from the fact that given $W_i$, the probability of a symbol getting erased is independent of anything else (even the input symbols). The equality in \eqref{eq:erasure4} follows
	because the event  $\{Y_i\neq \mc{E}:i \not\in J^{\mc{E}}\}$ contains the event $\{Y_i: i \not\in J^{\mc{E}}\}$.

	Note that we need to analyze only the second term in \eqref{eq:erasure5}. Since the first term divided by $n$, as argued in the proof of Lemma~\ref{lem:HUbarYgivenXandWErasure}, has an almost sure limit $\EX_\pi[h(p(W))].$
	
	Note that for an erasure channel, if $Y_i$ is not an erasure, $Y_i$ has the same value as that of 
	$X_i$ and $\{X_i\}$ are chosen independently of $\{W_i\}$. So, for 
	any joint distribution $\PR_{\mbf{X}}$ of input symbols:
	\begin{align}
	& -\frac{1}{n} \log \PR(\{Y_i: i \not\in J^{\mc{E}}\}|\{Y_i\neq \mc{E}, W_i: i \not\in J^{\mc{E}}\})  \nonumber \\
	& = - \frac{1}{n} \log \PR_{\mbf{X}}(\{Y_i: i \not\in J^{\mc{E}}\}) \nonumber \\
	& = - \frac{n-|J^{\mc{E}}|}{n} \frac{1}{n-|J^{\mc{E}}|} \log \PR_{\mbf{X}}(\{Y_i: i \not\in J^{\mc{E}}\}).\nonumber
	\end{align}
	
	By ergodicity  $\frac{|J^{\mc{E}}|}{n}$ converges a.s. to $\EX_\pi[p(W)] < 1$. So,  a.s. $n-|J^{\mc{E}}| \to \infty$. Thus, a.s. 
	$\lim_{n\to\infty} -\frac{1}{n-|J^{\mc{E}}|} \log \PR_{\mbf{X}}(\{Y_i: i \not\in J^{\mc{E}}\})$  is upper-bounded by the entropy rate of an i.i.d. Bernoulli$(0.5)$ process $\mbf{X}$, i.e., 
	$1$. Thus, a.s. the second term in \eqref{eq:erasure5} is  at most $1-\EX_\pi[p(W)]$.
\end{IEEEproof}

\begin{IEEEproof}[Proof of Lemma~\ref{lem:HUbarYgivenWErasureAchieve}]
	This proof is similar to the proof of Lemma~\ref{lem:HUbarYgivenWErasure}, except the fact that
	$$\lim_{n\to\infty} -\frac{1}{n-|J^{\mc{E}}|} \log \PR_{\mbf{X}}(\{Y_i: i \not\in J^{\mc{E}}\})=1$$ when
	$\PR_{\mbf{X}}$ is  i.i.d. uniform.
\end{IEEEproof}

\bibliographystyle{IEEEtran}
\newcommand{\SortNoop}[1]{}

\vspace{-1cm}
\begin{IEEEbiographynophoto}{Prabha Mandayam} is an Assistant Professor at the Department of Physics,  IIT Madras. Her research interests lie in the areas of quantum error correction, quantum information and quantum cryptography. She obtained her Masters in Physics from IIT Madras and Ph.D. in Physics from the California Institute of Technology.
\end{IEEEbiographynophoto}\vspace{-1cm}
\begin{IEEEbiographynophoto}{Krishna Jagannathan} is an Associate Professor at the Department of Electrical Engineering, IIT Madras. He obtained his B. Tech. in Electrical Engineering from IIT Madras in 2004, and the S.M. and Ph.D. degrees in Electrical Engineering and Computer Science from the Massachusetts Institute of Technology in 2006 and 2010 respectively. His research interests lie in the areas of stochastic modeling and analysis of communication networks, online learning, network control, and queuing theory.
\end{IEEEbiographynophoto}\vspace{-1cm}
\begin{IEEEbiographynophoto}{Avhishek Chatterjee} is an Assistant Professor at the Department of Electrical Engineering, IIT Madras. He obtained his Ph.D. in Electrical and Computer Engineering from The University of Texas at Austin. His research interest is in stochastic networks and information dynamics.
	\end{IEEEbiographynophoto}

\end{document}